\documentclass[11pt, twoside, letterpaper]{amsart}
\usepackage{amsmath, hyperref, color}
\usepackage{enumerate}
\usepackage{setspace}
\usepackage{soul}

\usepackage[left=2.5cm,top=3cm,hcentering,vcentering]{geometry}

\newtheorem{thm}{Theorem}[section]
\newtheorem{cor}[thm]{Corollary}

\newtheorem{prop}[thm]{Proposition}

\theoremstyle{definition}
\newtheorem{defn}[thm]{Definition}

\theoremstyle{remark}
\newtheorem{rem}[thm]{Remark}

\numberwithin{equation}{section}
\newcommand{\R}{\mathbb{R}}
\newcommand{\N}{\mathbb{N}}

\newcommand{\ind}{\mathbf{1}}
\newcommand{\prob}{\mathbb{P}}

\newcommand{\FF}{\mathbb{F}}

\newcommand{\PP}{\mathbb{P}}
\newcommand{\QQ}{\mathbb{Q}}
\newcommand{\expec}{\mathbb{E}}
\newcommand{\cF}{\mathcal{F}}

\newcommand{\Xtilde}{\widetilde{X}}
\newcommand{\dbra}[1]{[\kern-0.15em[ #1 ]\kern-0.15em]}
\newcommand{\dbraco}[1]{[\kern-0.15em[ #1 [\kern-0.15em[}
\newcommand{\dbraoc}[1]{]\kern-0.15em] #1 ]\kern-0.15em]}
\newcommand{\dbraoo}[1]{]\kern-0.15em] #1 [\kern-0.15em[}

\newcommand{\be}{\begin{equation}}
\newcommand{\ee}{\end{equation}}
\newcommand{\ba}{\begin{aligned}}
\newcommand{\ea}{\end{aligned}}

\begin{document}


\title[Sure profits via flash strategies]{On the existence of sure profits via flash strategies}

\author[C. Fontana]{Claudio Fontana}
\address[Claudio Fontana]{Department of Mathematics ``Tullio Levi - Civita'', University of Padova, Italy.}
\email{fontana@math.unipd.it}

\author[M. Pelger]{Markus Pelger}
\address[Markus Pelger]{Management Science \& Engineering Department, Stanford University, Huang Engineering Center, Stanford, CA, United States}
\email{mpelger@stanford.edu}

\author[E. Platen]{Eckhard Platen}
\address[Eckhard Platen]{School of Mathematical and Physical Sciences and Finance Discipline Group, University of Technology Sydney, Broadway NSW 2007, Sydney, Australia, and Department of Actuarial Sciences, University of Cape Town.}
\email{Eckhard.Platen@uts.edu.au}

\thanks{The first author gratefully acknowledges the support of the Bruti-Liberati Visiting Fellowship and the hospitality of the Quantitative Finance Research Centre at the Finance Discipline Group at the University of Technology Sydney.}

\date{\today}


\maketitle

\begin{abstract}
We introduce and study the notion of sure profit via flash strategy, consisting of a high-frequency limit of buy-and-hold trading strategies. 
In a fully general setting, without imposing any semimartingale restriction, we prove that there are no sure profits via flash strategies if and only if asset prices do not exhibit predictable jumps.
This result relies on the general theory of processes and provides the most general formulation of the well-known fact that, in an arbitrage-free financial market, asset prices (including dividends) should not exhibit jumps of a predictable direction or magnitude at predictable times.
We furthermore show that any price process is always right-continuous in the absence of sure profits.
Our results are  robust  under small transaction costs  and imply that, under minimal assumptions, price changes occurring at scheduled dates should only be due to unanticipated information releases.
\end{abstract}

\section{Introduction}

In the financial markets literature, the importance of allowing for jumps in asset prices at scheduled or predictable dates is widely acknowledged. Indeed, asset prices move in correspondence of macroeconomic news announcements (see \cite{Evans,KV,KimWright,LeeMyk,Rangel}), publication of earnings reports (see \cite{DubJoh,LeeMyk}), dividend payments (see \cite{HJ88}), Federal Reserve meetings (see \cite{piaz01,Piazzesi}), major political decisions\footnote{Recent political events like the Brexit and the election of the American president in 2016 represent striking examples of the impact on financial markets of discontinuities happening at scheduled or predictable dates.}, and all these events take place at dates which are typically known in advance.
In the context of continuous-time models, \cite{Lee} reports significant empirical evidence on jump predictability, while a model of the US Treasury rate term structure with jumps occurring in correspondence of employment report announcement dates is developed in \cite{KimWright}
(see also \cite{GS16,FS18} in the case of credit risky term structures).
Hence, realistic financial models should account for the presence of jumps at predictable times.

According to the efficient market hypothesis, asset prices should fully reflect all available information (see \cite{Fama}). In particular, if asset prices suddenly change at  scheduled or predictable dates, then this can only be due to the release of unanticipated information. Indeed, under market efficiency, if the released information does not contain any surprise element, then it should be already incorporated in market prices and, hence, prices should not move. 
This implication of market efficiency is coherent with absence of arbitrage: if a price process is known to jump at a given point in time, then the direction and the size of the jump should not be perfectly known in advance, otherwise arbitrage profits would be possible. As pointed out in \cite[Section 2.1]{ABD}, this can be easily understood by analogy to discrete-time models, where absence of arbitrage implies that the return over each single trading period can never be predicted.
Summing up, market efficiency and absence of arbitrage suggest that asset prices cannot exhibit {\em predictable jumps}, i.e., discontinuities such that the time of the jump {\em and} the direction (or even the exact magnitude) of the jump can be known in advance.
 
The goal of the present paper is to characterize the minimal no-arbitrage condition under which asset prices do not exhibit predictable jumps.
We work in a general stochastic model of a financial market and we refrain from imposing any assumption on the price process, except for mild path regularity. In particular, we do not assume the semimartingale property, relying instead on fundamental tools from the general theory of processes. 
We only allow for realistic trading strategies, consisting of bounded buy-and-hold positions and high-frequency limits thereof, which we name {\em flash strategies} (Definition \ref{def:flash}). 
Our central result (Theorem \ref{thm:main}) shows that the existence of predictable jumps is equivalent to the possibility of realizing sure profits via flash strategies, and even constant profits if also the size of the jump can be predicted. 
In the semimartingale case, these sure profits can be realized instantaneously (Corollary \ref{cor:semimg}).
We furthermore show that right-continuity is an indispensable requirement in order to exclude constant profits from flash strategies (Section \ref{sec:RC}).
Since constant profits persist under small transaction costs (Section \ref{sec:robust}), this provides a sound justification for the ubiquitous assumption of right-continuity in mathematical finance.
From the probabilistic standpoint, our approach sheds new light on path properties of stochastic processes, linking them to economically meaningful no-arbitrage requirements.

This study is motivated by the possibility of arbitrage in high-frequency markets. In particular, our notion of a flash strategy is similar to a {\em directional event-based strategy} (see \cite[Chapter 9]{Aldridge}). Such strategies aim at realizing positive profits in correspondence of some predetermined market events. In the case of anticipated events, such as scheduled macroeconomic announcements, the strategy is opened ahead of the event and liquidated just after the event. The holding period is typically very short and the speed of response determines the trade gain. 
Our notion of flash strategy can also represent a {\em latency arbitrage strategy} (see  \cite[Chapter 12]{Aldridge}): if the same asset is traded in two markets at slightly different prices, then high-frequency traders can arbitrage the price difference by simultaneously trading in the two markets. Since our price process is allowed to be multi-dimensional, this situation can be easily captured by representing the prices of the same asset on different markets as different components of a vector price process.
Other kinds of high-frequency strategies that can be represented via flash strategies include {\em front-running strategies}, as described in the best-selling book \cite{Lewis} (see also \cite{Pro_review}).
Our results indicate that the existence of predictable jumps lies at the origin of the sure profits generated by these types of high-frequency arbitrage strategies.
As shown in the recent empirical analysis of \cite{Ted} on the Eurex option market, sure profits via flash strategies can occur in financial markets (see Remark \ref{rem:Ted}).

The possibility of sure (or even constant) profits generated by predictable jumps is also related to the classical issue of the behavior of ex-dividend prices at dividend payment dates, as considered in \cite{HJ88} (see also \cite{Battauz}). Typically, the dividend payment date and the amount of the dividend are known in advance (i.e., they are predictable). \cite{HJ88} show that, if there exists a martingale measure, then either the ex-dividend price drops exactly by the amount of the dividend or the jump in the ex-dividend price cannot be predictable. 
In this perspective, our results can be regarded as the most general formulation of the seminal result of \cite{HJ88} (to this effect, see Remark \ref{rem:escrowed_div}).

The rest of the paper is organized as follows. Section \ref{sec:setting} introduces the probabilistic setting. The class of trading strategies under consideration is defined in Section \ref{sec:strategies}. Section \ref{sec:main} contains our central result, characterizing predictable jumps in terms of sure profits via flash strategies.
The role of right-continuity and the robustness of sure profits via flash strategies are analysed respectively in Sections \ref{sec:RC} and \ref{sec:robust}, while the semimartingale case is studied in Section \ref{sec:semimg}. We discuss the relations with other no-arbitrage conditions in Section \ref{sec:relations} and we then conclude in Section \ref{sec:conclusions}.

\section{Sure and constant profits via flash strategies}

\subsection{Setting}	\label{sec:setting}

Let $(\Omega,\cF,\PP)$ be a probability space endowed with a filtration $\FF=(\cF_t)_{t\geq0}$ satisfying the usual conditions of right-continuity and completeness and supporting a c\`adl\`ag (right-continuous with left limits) real-valued\footnote{We restrict our presentation to the case of a one-dimensional process $X$ for clarity of notation only. The multi-dimensional case is completely analogous and can be treated with the same tools.}  adapted process $X=(X_t)_{t\geq0}$. 
The filtration $\FF$ represents the flow of available information, while the process $X$ represents the gains process of a risky asset, discounted with respect to some baseline security. In the case of a dividend paying asset, this corresponds to the sum of the discounted ex-dividend price and the cumulated discounted dividends.
We do not assume that $X$ is a semimartingale nor that the initial sigma-field $\cF_0$ is trivial. The results presented below apply to any model in a finite time horizon $T<+\infty$ by simply considering the stopped process $X^T$.
We denote by $\Delta X=(\Delta X_t)_{t\geq0}$ the jump process of $X$, with $\Delta X_t:=X_t-X_{t-}$, for $t\geq0$. Following the convention of \cite{MR1943877}, we let $\Delta X_0=0$.
We refer to \cite{MR1943877} for all unexplained notions related to the general theory of stochastic processes.

A stopping time $T$ is said to be a {\em jump time} of $X$ if $\dbra{T}\subseteq\{\Delta X\neq0\}$ (up to an evanescent set)\footnote{We recall that the graph of a stopping time $T$ is defined as $\dbra{T}=\{(\omega,t)\in\Omega\times\R_+ : T(\omega)=t\}$. Similarly, for two stopping times $\sigma$ and $\tau$, we can define the stochastic interval $\dbraoc{\sigma,\tau}=\{(\omega,t)\in\Omega\times\R_+:\sigma(\omega)<t\leq\tau(\omega)\}$.}. 
We say that $X$ exhibits {\em predictable jumps} if there exists at least one jump time $T$ which is a predictable time and such that the random variable $\ind_{\{T<+\infty,\,\Delta X_T>0\}}$ is $\cF_{T-}$-measurable.
Strengthening this definition, we say that $X$ exhibits {\em fully predictable jumps} if there exists at least one predictable jump time $T$ such that the random variable $\Delta X_T\ind_{\{T<+\infty\}}$ is $\cF_{T-}$-measurable.
In other words, $X$ exhibits predictable jumps if there exists at least one predictable jump time at which the direction (and even the size, in the case of a fully predictable jump) of the jump is known just before the occurrence of the jump.
We aim at relating the absence of predictable jumps (and of fully predictable jumps) to minimal and realistic no-arbitrage properties.

\subsection{Buy-and-hold strategies and flash strategies}	\label{sec:strategies}

We describe the activity of trading in the financial market according to the following definition. 

\begin{defn}	\label{def:buy_and_hold}
A {\em buy-and-hold strategy} is a stochastic process $h$ of the form $h=\xi\ind_{\dbraoc{\sigma,\tau}}$, where $\sigma$ and $\tau$ are two bounded stopping times such that $\sigma\leq\tau$ a.s. and $\xi$ is a bounded $\cF_{\sigma}$-measurable random variable.
\end{defn}

A buy-and-hold strategy corresponds to the simplest possible trading strategy: a portfolio $\xi$ is formed at time $\sigma$ and liquidated at time $\tau$. Note that the portfolio $\xi$ is restricted to be bounded, thus excluding arbitrarily large positions in the traded assets.
For a buy-and-hold strategy $h$, the gains from trading at date $t$ are given by $(h\cdot X)_t:=\xi(X_{\tau\wedge t}-X_{\sigma\wedge t})$, for $t\geq0$.

\begin{defn}	\label{def:flash}
A {\em flash strategy} is a sequence $(h^n)_{n\in\N}$ of buy-and-hold strategies $h^n=\xi^n\ind_{\dbraoc{\sigma_n,\tau_n}}$ such that the random variables $(\xi^n)_{n\in\N}$ are bounded uniformly in $n$ and the following two properties hold a.s. for $n\rightarrow+\infty$:
\begin{enumerate}[(i)]
\item the sequences $(\sigma_n)_{n\in\N}$ and $(\tau_n)_{n\in\N}$ converge to some stopping time $\tau$ with $\PP(\tau<+\infty)>0$;
\item the random variables $(\xi^n)_{n\in\N}$ converge to some random variable $\xi$. 
\end{enumerate}
A flash strategy $(h^n)_{n\in\N}$ is said to generate a {\em sure profit} if $(h^n\cdot X)_t$ converges a.s. to $\zeta\ind_{\{\tau\leq t\}}$, for all $t\geq0$, for some random variable $\zeta$ such that $\{\tau<+\infty\}\subseteq\{\zeta>0\}$.
If $\PP(\zeta=c)=1$, for some constant $c>0$, then the flash strategy $(h^n)_{n\in\N}$ is said to generate a {\em constant profit}.
\end{defn}

A flash strategy represents the possibility of investing at higher and higher frequencies. In the limit, the strategy converges to a (bounded) position $\xi$ which is constructed and then immediately liquidated at some random time $\tau$. If by doing so and starting from zero initial wealth an investor can reach a strictly positive amount of wealth (provided that the investor trades at all, i.e., from time $\tau$ onwards), then the flash strategy is said to generate a {\em sure profit}.
In the case of a {\em constant profit}, the amount of wealth generated by the flash strategy is perfectly known in advance.
The requirement that the positions $(\xi^n)_{n\in\N}$ are uniformly bounded means that an investor is not allowed to make larger and larger trades as the holding period $\tau_n-\sigma_n$ converges to zero. 
This makes flash strategies feasible by placing market orders in financial markets with finite liquidity.
Observe also that no trading activity occurs  in the limit on the event $\{\tau=+\infty\}$.

In the limit, a sure profit does not involve any risk, since the gains from trading converge to a strictly positive random variable. Moreover, it turns out that the components $(h^n)_{n\in\N}$ of a flash strategy generating a sure profit can be chosen in such a way that the potential losses incurred by {\em each} individual buy-and-hold strategy $h^n$ are uniformly bounded, for all sufficiently large $n$ (see Section \ref{sec:main}).
Note also that, if a flash strategy generates a constant profit for {\em some} $c>0$, then there exist constant profits for {\em every} $c>0$, since the flash strategy can be arbitrarily rescaled.
A further important property of the notion of constant profit via flash strategies is its robustness with respect to small transaction costs (see Section \ref{sec:robust} below).


\begin{rem}	\label{rem:Ted}
Sure, and even constant, profits via flash strategies can occur in financial markets. For instance, in a recent empirical analysis of the Eurex option market, \cite{Ted} demonstrates the existence of arbitrage strategies consisting of two opposed market orders (i.e., buy and sell) executed within a time window of less than three seconds and leading to riskless immediate gains.
Such strategies are shown to be profitable for market makers, who face reduced transaction fees (to this effect, see also Section \ref{sec:robust}).
\end{rem}

\subsection{Predictable jumps and sure profits via flash strategies}	\label{sec:main}

The following theorem shows that the absence of sure profits via flash strategies is equivalent to the absence of predictable jumps.
This result relies on the fact that predictable jumps are anticipated by a sequence of precursory signals which can be used to construct a sequence of buy-and-hold strategies forming a flash strategy.

\begin{thm} 	\label{thm:main}
The process $X$ does not exhibit predictable (fully predictable, resp.) jumps if and only if there are no sure (constant, resp.) profits via flash strategies.
\end{thm}

\begin{proof}
We first prove that if $X$ exhibits predictable jumps, then there exist sure profits. To this effect, let $T$ be a predictable time with $\dbra{T}\subseteq\{\Delta X\neq0\}$ such that the random variable $\ind_{\{T<+\infty,\,\Delta X_T>0\}}$ is $\cF_{T-}$-measurable. For simplicity of notation, we set $\Delta X_{T}=0$ on $\{T=+\infty\}$.
In view of \cite[Theorem I.2.15]{MR1943877}, there exists an announcing sequence $(\rho_n)_{n\in\N}$ of stopping times satisfying $\rho_n<T$ and such that $\rho_n$ increases to $T$ for $n\rightarrow+\infty$.
For each $n\in\N$, let $\sigma_n:=\rho_n\wedge n$ and $\tau_n:=T\wedge n$ and define the sequence $(h^n)_{n\in\N}$ by 
\be	\label{eq:constr_thm_sure}
h^n = \xi^n\ind_{\dbraoc{\sigma_n,\tau_n}},
\qquad\text{ where }\;
\xi^n := 2\,\PP(\Delta X_T>0|\cF_{\sigma_n})-1,
\qquad\text{ for every }n\in\N.
\ee
As a consequence of the martingale convergence theorem, the sequence $(\xi^n)_{n\in\N}$ converges a.s. to the random variable
\[
\xi := 2\,\PP(\Delta X_T>0|\cF_{T-})-1
= \ind_{\{\Delta X_T>0\}} - \ind_{\{\Delta X_T\leq0\}},
\]
where we have used the fact that $\ind_{\{\Delta X_T>0\}}$ is $\cF_{T-}$-measurable.
This shows that $(h^n)_{n\in\N}$ is a flash strategy in the sense of Definition \ref{def:flash}. 
To prove that it generates a sure profit, it suffices to remark that, for every $t\geq0$, it holds that $\lim_{n\rightarrow+\infty}X_{\tau_n\wedge t}=X_{T\wedge t}$ and
\[
\lim_{n\rightarrow+\infty}X_{\sigma_n\wedge t}
= X_{T-}\ind_{\{T\leq t\}} + X_t\ind_{\{T>t\}},
\]
so that
\[
\lim_{n\rightarrow+\infty}(h^n\cdot X)_t 
= \lim_{n\rightarrow+\infty}\bigl(\xi^n(X_{\tau_n\wedge t}-X_{\sigma_n\wedge t})\bigr)
= \xi\Delta X_{T}\ind_{\{T\leq t\}}
= |\Delta X_T|\ind_{\{T\leq t\}}
\quad\text{ a.s.},
\]
thus showing that $(h^n)_{n\in\N}$ generates a sure profit.

We now turn to the converse implication. Let $(h^n)_{n\in\N}$ be a flash strategy, composed of elements of the form $h^n=\xi^n\ind_{\dbraoc{\sigma_n,\tau_n}}$, generating a sure profit with respect to a random variable $\zeta$ and a stopping time $\tau$ with $\{\tau<+\infty\}\subseteq\{\zeta>0\}$.
It can be checked that $\lim_{k\rightarrow+\infty}(h^n\cdot X)_{t-1/k}=(h^n\cdot X)_{t-}$ uniformly over $n\in\N$. Indeed, defining $\bar{\xi}:=\sup_{n\in\N}|\xi^n|$ (which is a bounded random variable due to  Definition \ref{def:flash}), it holds that
\begin{align*}
& \lim_{k\rightarrow+\infty}\sup_{n\in\N}\,\bigl|(h^n\cdot X)_{t-\frac{1}{k}}-(h^n\cdot X)_{t-}\bigr|
= \lim_{k\rightarrow+\infty}\sup_{n\in\N}\,\bigl|\xi^n\bigl(X^{\tau_n}_{t-\frac{1}{k}}-X^{\sigma_n}_{t-\frac{1}{k}}\bigr)-\xi^n\bigl(X^{\tau_n}_{t-}-X^{\sigma_n}_{t-}\bigr)\bigr|	\\
&\quad \leq \bar{\xi} \lim_{k\rightarrow+\infty}\Bigl(
\sup_{n\in\N}\ind_{\{t-\frac{1}{k}<\tau_n<t\}}|X_{\tau_n}-X_{t-\frac{1}{k}}|
+\sup_{n\in\N}\ind_{\{t-\frac{1}{k}<\sigma_n<t\}}|X_{\sigma_n}-X_{t-\frac{1}{k}}|\Bigr)\\
&\quad \leq 2\,\bar{\xi} \lim_{k\rightarrow+\infty}\sup_{u\in(t-\frac{1}{k},t)}|X_u-X_{t-}| = 0.
\end{align*}
Hence, by the Moore-Osgood theorem, we can conclude that, for every $t\geq0$,
\begin{equation}	\label{eq:jump_sure}
\zeta\ind_{\{t=\tau\}}
= \lim_{n\rightarrow+\infty}(h^n\cdot X)_t - \lim_{k\rightarrow+\infty}\lim_{n\rightarrow+\infty}(h^n\cdot X)_{t-\frac{1}{k}}
= \Delta X_t\,\bar{h}_t
\qquad\text{a.s.},
\end{equation}
with $\bar{h}_t:=\lim_{n\rightarrow+\infty}h^n_t=\lim_{n\rightarrow+\infty} \xi^n\ind_{\{\sigma_n<t\leq\tau_n\}}$, for all $t\geq0$.
Letting $\xi=\lim_{n\rightarrow+\infty}\xi^n$ (see Definition \ref{def:flash}), a first implication of \eqref{eq:jump_sure} is that $\{\tau<+\infty\}\subseteq\{\xi\neq0\}$ and $\dbra{\tau}\subseteq\{\Delta X\neq0\}$, up to an evanescent set, so that $\tau$ is a jump time of $X$. 
Furthermore, always by \eqref{eq:jump_sure}, on $\{\tau<+\infty\}$ it holds that $\{\Delta X_{\tau}>0\}=\{\bar{h}_{\tau}>0\}$.
Noting that the random variables $\xi^n\ind_{\{\sigma_n<\tau\}}$ and $\ind_{\{\tau\leq\tau_n\}}$ are $\cF_{\tau-}$-measurable for every $n\in\N$ (see e.g. \cite[\textsection~I.1.17]{MR1943877}), this implies that $\ind_{\{\tau<+\infty,\,\Delta X_{\tau}>0\}}$ is $\cF_{\tau-}$-measurable as well.
To complete the proof, it remains to show that $\tau$ is a predictable time.
For each $n\in\N$, let $A_n:=\{\sigma_n<\tau\leq\tau_n\}\cap\{\xi^n\neq0\}$ and note that $A_n\subseteq\{\tau<+\infty\}$, since each stopping time $\tau_n$ is bounded.
Moreover, it holds that
\[
\lim_{n\rightarrow+\infty}\ind_{A_n}
= \lim_{n\rightarrow+\infty}\ind_{\{\sigma_n<\tau\leq\tau_n\}}\xi^n
\frac{\ind_{\{\xi^n\neq0\}}}{\xi^n}
= \frac{\zeta}{\xi\Delta X_{\tau}}\ind_{\{\tau<+\infty\}} 
\quad\text{ a.s. }
\]
This identity shows that the sequence $(A_n)_{n\in\N}$ is convergent, with $\lim_{n\rightarrow+\infty}A_n=\{\tau<+\infty\}$ and $\xi\Delta X_{\tau}=\zeta$ on $\{\tau<+\infty\}$ (up to a $\PP$-nullset).
Since the stopping times $(\sigma_n)_{n\in\N}$ and $(\tau_n)_{n\in\N}$ converge a.s. to $\tau$ for $n\rightarrow+\infty$, this implies that $\dbra{\tau}\subseteq\liminf_{n\rightarrow+\infty}\dbraoc{\sigma_n,\tau_n}\subseteq\limsup_{n\rightarrow+\infty}\dbraoc{\sigma_n,\tau_n}\subseteq\dbra{\tau}$, so that $\dbra{\tau}=\lim_{n\rightarrow+\infty}\dbraoc{\sigma_n,\tau_n}$. Since each stochastic interval $\dbraoc{\sigma_n,\tau_n}$ is a predictable set (see e.g. \cite[Proposition I.2.5]{MR1943877}), it follows that $\dbra{\tau}$ is also a predictable set, i.e., $\tau$ is a predictable time.

Let us now prove that $X$ exhibits fully predictable jumps if and only if there exists a flash strategy generating a constant profit, following a similar line of reasoning as in the first part of the proof.
Let $T$ be a predictable time with $\dbra{T}\subseteq\{\Delta X\neq0\}$ such that the random variable $\Delta X_{T}\ind_{\{T<+\infty\}}$ is $\cF_{T-}$-measurable. 
Fix some constant $k\geq1$ such that $\PP(T<+\infty,|\Delta X_{T}|\in[1/k,k])>0$ and define the stopping time $\tau:=T\ind_{\{|\Delta X_{T}|\in[1/k,k]\}}+\infty\ind_{\{|\Delta X_{T}|\notin[1/k,k]\}}$. 
By \cite[Proposition I.2.10]{MR1943877}, $\tau$ is a predictable time and, therefore, there exists an announcing sequence $(\rho_n)_{n\in\N}$ of stopping times satisfying $\rho_n<\tau$ and such that $\rho_n$ increases to $\tau$ for $n\rightarrow+\infty$.
Similarly as in the first part of the proof, let $\sigma_n:=\rho_n\wedge n$ and $\tau_n:=\tau\wedge n$, for each $n\in\N$, and define the sequence $(h^n)_{n\in\N}$ by 
\be	\label{eq:constr_thm_constant}
h^n = \xi^n\ind_{\dbraoc{\sigma_n,\tau_n}},
\qquad\text{ where }\;
\xi^n := k\frac{\bigl|\expec[\Delta X_{\tau}|\cF_{\sigma_n}]\bigr|\wedge\frac{1}{k}}{\expec[\Delta X_{\tau}|\cF_{\sigma_n}]},
\qquad\text{ for every }n\in\N,
\ee
with the conventions $\Delta X_{\tau}=0$ on $\{\tau=+\infty\}$ and $\frac{0}{0}=0$.
By construction, it holds that $|\xi^n|\leq k$, for every $n\in\N$, so that $(h^n)_{n\in\N}$ is well-defined as a sequence of buy-and-hold strategies. 
Moreover, the sequence $(\xi^n)_{n\in\N}$ converges a.s. to the random variable
\[
\xi := 
k\frac{\bigl|\expec[\Delta X_{\tau}|\cF_{\tau-}]\bigr|\wedge\frac{1}{k}}{\expec[\Delta X_{\tau}|\cF_{\tau-}]}
= k\frac{\bigl|\Delta X_{\tau}\bigr|\wedge\frac{1}{k}}{\Delta X_{\tau}}
= \frac{\ind_{\{\Delta X_{\tau}\neq0\}}}{\Delta X_{\tau}},
\]
where the second equality makes use of the fact that $\Delta X_{\tau}$ is $\cF_{\tau-}$-measurable, as follows from the identity $\Delta X_{\tau}=\Delta X_{T}\ind_{\{|\Delta X_{T}|\in[1/k,k]\}}$ together with the $\cF_{T-}$-measurability of $\Delta X_{T}\ind_{\{T<+\infty\}}$.
This shows that $(h^n)_{n\in\N}$ is a flash strategy in the sense of Definition \ref{def:flash}. 
Moreover, it holds that
\[
\lim_{n\rightarrow+\infty}(h^n\cdot X)_t 
= \lim_{n\rightarrow+\infty}\bigl(\xi^n(X_{\tau_n\wedge t}-X_{\sigma_n\wedge t})\bigr)
= \xi\Delta X_{\tau}\ind_{\{\tau\leq t\}}
= \ind_{\{\tau\leq t\}}
\quad\text{ a.s.},
\]
thus showing that $(h^n)_{n\in\N}$ generates a constant profit with respect to $c=1$.

Conversely, let $(h^n)_{n\in\N}$ be a flash strategy, with $h^n=\xi^n\ind_{\dbraoc{\sigma_n,\tau_n}}$, $n\in\N$, generating a constant profit with respect to $c>0$ and a stopping time $\tau$.
Similarly as in the case of a sure profit, it holds that
$c = \Delta X_{\tau}\lim_{n\rightarrow+\infty}h^n_{\tau}$ a.s. on $\{\tau<+\infty\}$.
This implies that $\dbra{\tau}\subseteq\{\Delta X\neq0\}$ up to an evanescent set, so that $\tau$ is a jump time of $X$. Moreover, it holds that $\Delta X_{\tau}=c/(\lim_{n\rightarrow+\infty} \xi^n\ind_{\{\sigma_n<\tau\leq\tau_n\}})$ a.s. on $\{\tau<+\infty\}$, from which the $\cF_{\tau-}$-measurability of $\Delta X_{\tau}\ind_{\{\tau<+\infty\}}$  follows.
Finally, the predictability of $\tau$ can be shown exactly as above in the case of a sure profit.
\end{proof}

\begin{rem}	\label{rem:escrowed_div}
Examples of models allowing for constant profits via flash strategies are given by the {\em escrowed dividend models} introduced in \cite{Roll,Geske,Whaley} (see also the analysis in \cite{HJ88}). Indeed, such models consider an asset paying a deterministic dividend at a known date and assume that the ex-dividend price drops by a fixed fraction $\delta\in(0,1)$ of the dividend at the dividend payment date. This corresponds to a fully predictable jump of the process $X$ and, hence, in view of Theorem \ref{thm:main}, can be exploited to generate a constant profit via a flash strategy.
\end{rem}

It is important to remark that, although a flash strategy $(h^n)_{n\in\N}$ generating a sure profit does not involve any risk in the limit, each individual buy-and-hold strategy $h^n$ carries the risk of potential losses. However, a flash strategy can be constructed in such a way that losses are uniformly bounded, as we are going to show in the remaining part of this section. 
This is an important property of flash strategies, especially in view of their practical applicability.

By Theorem \ref{thm:main}, there are sure profits via flash strategies if and only if $X$ exhibits predictable jumps. Hence, let $T$ be a predictable time with $\dbra{T}\subseteq\{\Delta X\neq0\}$ such that $\ind_{\{T<+\infty,\,\Delta X_T>0\}}$ is $\cF_{T-}$-measurable.
Consider the event $A(N,C):=\{T\leq N,|X_{T-}|\leq C,\Delta X_T>0\}\in\cF_{T-}$, for some constants $N>0$ and $C\geq1$ such that $\prob(A(N,C))>0$, and define the predictable time $\tau:=T\ind_{A(N,C)}+\infty\ind_{A(N,C)^c}$.
Define then the sequences of stopping times $(\sigma_n)_{n\in\N}$ and $(\tau_n)_{n\in\N}$  by $\sigma_n:=\rho_n\wedge n$ and $\tau_n:=\tau\wedge n$, for each $n\in\N$, where $(\rho_n)_{n\in\N}$ is an announcing sequence for $\tau$. Similarly as in \eqref{eq:constr_thm_sure}, we construct the buy-and-hold strategy $h^n=\xi^n\ind_{\dbraoc{\sigma_n,\tau_n}}$, with
\be	\label{eq:strategy_loss}
\xi^n := 
\frac{\PP(\tau<+\infty|\cF_{\sigma_n})}{1+(|X_{\sigma_n}|-C)^+},
\qquad\text{ for every }n\in\N.
\ee
Since $X_{\sigma_n}\rightarrow X_{\tau-}$ a.s. for $n\rightarrow+\infty$ and $|X_{\tau-}|\leq C$ on $\{\tau<+\infty\}$, the sequence $(\xi^n)_{n\in\N}$ converges a.s. to $\xi=\ind_{\{\tau<+\infty\}}$.
The same arguments given in the first part of the proof of Theorem \ref{thm:main} allow then to show that $(h^n)_{n\in\N}$ generates a sure profit at $\tau$.
Furthermore, on $\{\tau<+\infty\}$, for every $n\in\N$ such that $n\geq N$, it holds that 
\begin{align*}
(h^n\cdot X)_{\tau} = \xi^n\left(X_{\tau}-X_{\rho_n}\right)
&= \xi^n\left(\Delta X_{\tau}+X_{\tau-}-X_{\rho_n}\right)	
&\geq -C-\frac{|X_{\rho_n}|}{1+(|X_{\rho_n}|-C)^+}
\geq -2C
\qquad\text{ a.s.}
\end{align*}
We have thus shown that, even if each individual buy-and-hold strategy $h^n$ does involve some risk, the potential losses from trading are uniformly bounded on  $\{\tau<+\infty\}$ for all sufficiently large $n$.
An analogous result can be shown to hold true in the case of flash strategies generating constant profits, modifying the strategy \eqref{eq:constr_thm_constant} in analogy to \eqref{eq:strategy_loss}.

\begin{rem}[On short-selling constraints]	\label{rem:short_sale}
The fact that predictable jumps lead to sure profits via flash strategies is robust with respect to the introduction of short-selling constraints, unless the predictable jumps of $X$ are a.s. negative.
This simply follows by noting that if $T$ is a predictable time with $\dbra{T}\subseteq\{\Delta X\neq0\}$ such that $\ind_{\{T<+\infty,\,\Delta X_T>0\}}$ is $\cF_{T-}$-measurable and $\prob(\Delta X_T>0)>0$, then in the first part of the proof of Theorem \ref{thm:main} the flash strategy $(h^n)_{n\in\N}$ can be chosen to consist of long positions in the asset,  as in the case of \eqref{eq:strategy_loss}.
Up to a suitable definition of the predictable time $\tau$, a similar reasoning applies to flash strategies generating constant profits.
\end{rem}

\section{Further properties and ramifications}		\label{sec:further}

In this section, we study some further properties of the notions of sure and constant profits via flash strategies. 
We first prove the necessity of the requirement of right-continuity for a price process $X$.
We then discuss the  behavior of the notion of sure profit via flash strategies  under small transaction costs. 
Finally, we specialize our results to the semimartingale case and discuss the relations with other no-arbitrage conditions.

\subsection{Right-continuity and sure profits}	\label{sec:RC}

As explained in Section \ref{sec:setting}, the process $X$ is allowed to be fully general, up to the mild requirement of path regularity, in the sense of right-continuity and existence of limits from the left. One might wonder whether right-continuity can be relaxed, assuming only that $X$ has l\`adl\`ag paths (i.e., with finite limits from the left and from the right).
As shown below, this is unfeasible, because right continuity represents an indispensable requirement for any arbitrage-free price process.
For a l\`adl\`ag process $X=(X_t)_{t\geq0}$, we denote by $X_{t+}$ the right-hand limit at $t$ and $\Delta^+X_t:=X_{t+}-X_t$, for $t\geq0$.

\begin{prop}	\label{prop:RC}
Assume that the process $X$ is l\`adl\`ag.
If $X$ fails to be right-continuous, then there exists a flash strategy $(h^n)_{n\in\N}$ such that $(h^n\cdot X)_t$ converges a.s. to $\ind_{\{\tau<t\}}$, for all $t\geq0$, for some  stopping time $\tau$ with $\prob(\tau<+\infty)>0$.
Conversely, if there exists such a flash strategy, then $X$ cannot be right-continuous.
\end{prop}
\begin{proof}
The argument is similar to the second part of the proof of Theorem \ref{thm:main}. Suppose that there exists a stopping time $T$ such that $\dbra{T}\subseteq\{\Delta^+X_T\neq0\}$. 
Fix a constant $k$ such that $\prob(T<+\infty,|\Delta^+X_T|\in[1/k,k])>0$ and define 
$
\tau:=T\ind_{\{|\Delta^+X_T|\in[1/k,k]\}}+\infty\ind_{\{|\Delta^+X_T|\notin[1/k,k]\}},
$ 
setting $\Delta^+X_T=0$ on $\{T=+\infty\}$.
Since the filtration $\FF$ is right-continuous, $\tau$ is a stopping time. Let define the sequences of bounded stopping times $(\sigma_n)_{n\in\N}$ and $(\tau_n)_{n\in\N}$ by
\[
\sigma_n := \tau\wedge n
\qquad\text{ and }\qquad
\tau_n := (\tau+n^{-1})\wedge n,
\qquad\text{ for each }n\in\N.
\]
It holds that $\bigvee_{n\in\N}\cF_{\sigma_n}=\cF_{\tau}$. Indeed, for any $A\in\cF_{\tau}$, define the sets 
\[
A_1:=A\cap\{\tau=\sigma_1\},
\quad
A_n:=\bigcap_{j<n}(A\cap\{\tau=\sigma_n\}\cap\{\tau>j\})
\quad\text{ and} \quad
A_{\infty}:=A\cap\{\tau=+\infty\}.
\] 
It can be checked that $A_{\infty}\in\cF_{\tau-}=\bigvee_{n\in\N}\cF_{\sigma_n-}\subseteq\bigvee_{n\in\N}\cF_{\sigma_n}$ and $A_n\in\cF_{\sigma_n}$, for each $n\in\N$.
Since $A=A_{\infty}\bigcup(\cup_{n=1}^{+\infty}A_n)$, this shows that $\cF_{\tau}\subseteq\bigvee_{n\in\N}\cF_{\sigma_n}$.
On the contrary, since $\sigma_n\leq\tau$, for every $n\in\N$, the inclusion $\bigvee_{n\in\N}\cF_{\sigma_n}\subseteq\cF_{\tau}$ is obvious.
Define now the sequence of buy-and-hold strategies $(h^n)_{n\in\N}$ by $h^n:=\xi^n\ind_{\dbraoc{\sigma_n\tau_n}}$, where
\[
\xi^n := k\frac{\bigl|\expec[\Delta^+X_{\tau}|\cF_{\sigma_n}]\bigr|\wedge\frac{1}{k}}{\expec[\Delta^+X_{\tau}|\cF_{\sigma_n}]},
\qquad\text{ for every }n\in\N.
\]
By the martingale convergence theorem, the random variables $(\xi_n)_{n\in\N}$ converge a.s. to
\[
\xi:=
k\frac{\bigl|\expec[\Delta^+X_{\tau}|\bigvee_{n\in\N}\cF_{\sigma_n}]\bigr|\wedge\frac{1}{k}}{\expec[\Delta^+X_{\tau}|\bigvee_{n\in\N}\cF_{\sigma_n}]}
= k\frac{\bigl|\expec[\Delta^+X_{\tau}|\cF_{\tau}]\bigr|\wedge\frac{1}{k}}{\expec[\Delta^+X_{\tau}|\cF_{\tau}]}
= \frac{\ind_{\{\Delta^+X_{\tau}\neq0\}}}{\Delta^+X_{\tau}},
\]
where we have used the right-continuity of the filtration $\FF$. Observe that $(h^n)_{n\in\N}$ is a flash strategy in the sense of Definition \ref{def:flash}. Moreover, for every $t\geq0$, it holds that $\lim_{n\rightarrow+\infty}X_{\sigma_n\wedge t}=X_{\tau\wedge t}$ and
\[
\lim_{n\rightarrow+\infty}X_{\tau_n\wedge t} = X_{\tau+}\ind_{\{\tau<t\}} + X_t\ind_{\{\tau\geq t\}},
\]
so that
\[
\lim_{n\rightarrow+\infty}(h^n\cdot X)_t
= \lim_{n\rightarrow+\infty}\bigl(\xi^n(X_{\tau_n\wedge t}-X_{\sigma_n\wedge t})\bigr)
= \xi\Delta^+X_{\tau}\ind_{\{\tau<t\}}
= \ind_{\{\tau<t\}}
\qquad\text{ a.s.},
\]
thus proving the first part of the proposition.

To prove the converse implication, let $(h^n)_{n\in\N}$ be a flash strategy such that $(h^n\cdot X)_t\rightarrow \ind_{\{\tau<t\}}$ a.s. as $n\rightarrow+\infty$, for all $t\geq0$, for some stopping time $\tau$ with $\prob(\tau<+\infty)>0$. Then, a straightforward adaptation of the arguments given in the last part of the proof of Theorem \ref{thm:main} allows to show that $\Delta^+X_{\tau}\neq0$ a.s. on $\{\tau<+\infty\}$, thus proving the claim.
\end{proof}

Proposition \ref{prop:RC} shows that the failure of right-continuity  leads to a constant profit from a flash strategy that can be realized at any time at which the price process jumps from the right. 
This result depends crucially on the right-continuity of the filtration $\FF$, which implies that $\Delta^+X_t$ is known at time $t$, immediately before the occurrence of the jump. Therefore, by trading sufficiently fast and liquidating the position immediately after the jump from the right, a trader can take advantage of this information and realize a constant profit.
In this sense, right-continuity is an essential requirement for any arbitrage-free price process.

\subsection{Behavior of profits from flash strategies under transaction costs}	\label{sec:robust}
In practice, transaction costs and market frictions can affect significantly the  feasibility of trading strategies, thus limiting the profitability of arbitrage strategies.  In this section, we study the behavior of  sure and constant profits via flash strategies with respect to small transaction costs. To this effect, let us formulate the following definition (see \cite{GR15}).

\begin{defn}	\label{def:robustness}
For $\varepsilon>0$, two strictly positive processes $X=(X_t)_{t\geq0}$ and $\Xtilde=(\Xtilde_t)_{t\geq0}$ are said to be {\em $\varepsilon$-close} if 
\[
\frac{1}{1+\varepsilon} \leq \frac{\Xtilde_t}{X_t} \leq 1+\varepsilon
\qquad
\text{ a.s. for all }t\geq0.
\]
\end{defn}

This definition corresponds to considering proportional transaction costs, with a bid (selling) price equal to $X_t/(1+\varepsilon)$ and an ask (buying) price equal to $X_t(1+\varepsilon)$. The definition also embeds the possibility of model mis-specifications, in the sense that the model price process $X$ corresponds to some true price process $\Xtilde$ up to a model error of magnitude $\varepsilon$.

In this context, assuming a strictly positive price process $X$, we shall say that sure/constant profits via flash strategies are  {\em robust} if they persist for every process $\Xtilde$ which is $\varepsilon$-close to $X$, for  sufficiently small $\varepsilon>0$.
This robustness property is made precise by the following proposition.

\begin{prop}	\label{prop:robust}
Assume that the process $X$ is strictly positive and admits constant profits via flash strategies. Then there exists a flash strategy $(h^n)_{n\in\N}$ and a predictable time $\tau$ such that, for every strictly positive process $\Xtilde$ which is $\varepsilon$-close to $X$, it holds that 
\be	\label{eq:robust}
\lim_{n\rightarrow+\infty}(h^n\cdot\Xtilde)_t\geq \bar{c}\,\ind_{\{\tau\leq t\}}
\qquad \text{ a.s. for all }t\geq0,
\ee
with $\bar{c}>0$, for sufficiently small $\varepsilon>0$.
\end{prop}
\begin{proof}
Suppose that $X$ admits a flash strategy $(\hat{h}^n)_{n\in\N}$ which generates a constant profit with respect to a stopping time $\hat{\tau}$. By Theorem \ref{thm:main}, $\hat{\tau}$ is a predictable time and $X$ exhibits a fully predictable jump at $\hat{\tau}$. Let $N>0$ be a constant such that $\prob(\hat{\tau}<+\infty,|X_{\hat{\tau}-}|\leq N)>0$ and define the predictable time $T:=\hat{\tau}\ind_{\{\hat{\tau}<+\infty,|X_{\hat{\tau}-}|\leq N\}}+\infty\ind_{\{\hat{\tau}=+\infty\}\cup\{\hat{\tau}<+\infty,|X_{\hat{\tau}-}|>N\}}$. 
Clearly, $X$ still exhibits a fully predictable jump at $T$. Hence, in view of Theorem \ref{thm:main}, there exists a flash strategy $(h^n)_{n\in\N}$, composed of elements $h^n=\xi^n\ind_{\dbraoc{\sigma_n,\tau_n}}$, $n\in\N$, with $|\xi^n|\leq k$ a.s. for all $n\in\N$, for some constant $k>0$, which generates a constant profit $c>0$ with respect to a predictable time $\tau$ with $\dbra{\tau}\subseteq\dbra{T}$ and $\prob(\tau<+\infty)>0$.
Let $\varepsilon>0$ and consider a strictly positive process  $\Xtilde$ which is $\varepsilon$-close to $X$.
Similarly as in \cite[Section 2]{CT15}, we can compute, for all $n\in\N$ and $t\geq0$:
\begin{align}
(h^n\cdot\Xtilde)_t 
&= \xi^n(\Xtilde_{\tau_n\wedge t}-\Xtilde_{\sigma_n\wedge t})	\notag\\
&\geq \xi^n\ind_{\{\xi^n\geq0\}}\left(\frac{X_{\tau_n\wedge t}}{1+\varepsilon}-(1+\varepsilon)X_{\sigma_n\wedge t}\right)
+  \xi^n\ind_{\{\xi^n<0\}}\left((1+\varepsilon)X_{\tau_n\wedge t}-\frac{X_{\sigma_n\wedge t}}{1+\varepsilon}\right)	\notag\\
&= \ind_{\{\xi^n\geq0\}}\frac{(h^n\cdot X)_t}{1+\varepsilon}+\ind_{\{\xi^n<0\}}(1+\varepsilon)(h^n\cdot X)_t
-|\xi^n|\varepsilon\frac{2+\varepsilon}{1+\varepsilon}X_{\sigma_n\wedge t}	\notag\\
&\geq  \ind_{\{\xi^n\geq0\}}\frac{(h^n\cdot X)_t}{1+\varepsilon}+\ind_{\{\xi^n<0\}}(1+\varepsilon)(h^n\cdot X)_t
-2\,\varepsilon k X_{\sigma_n\wedge t}.
\label{eq:proof_robust}
\end{align}
As shown in the first part of the proof of Theorem \ref{thm:main}, it holds that $X_{\sigma_n}\rightarrow X_{\tau-}$ a.s. on $\{\tau<+\infty\}$ for $n\rightarrow+\infty$. Hence, using Definition \ref{def:flash} and taking the limit for $n\rightarrow+\infty$ in \eqref{eq:proof_robust} yields
\[
\lim_{n\rightarrow+\infty}(h^n\cdot\Xtilde)_t 
\geq \ind_{\{\xi\geq0\}}\frac{c}{1+\varepsilon}+ \ind_{\{\xi<0\}}c(1+\varepsilon) - 2\,\varepsilon k X_{\tau-}
\geq \frac{c}{1+\varepsilon}-2\,\varepsilon k N
=: \bar{c}
\qquad\text{ a.s. on }\{\tau\leq t\}.
\]
For sufficiently small $\varepsilon$, it holds that $\bar{c}>0$. 
Furthermore, it can be easily verified that $(h^n\cdot\Xtilde)_t\rightarrow0$ a.s. on $\{\tau>t\}$ for $n\rightarrow+\infty$, thus proving the claim.
\end{proof}

The last proposition shows that the presence of fully predictable jumps represents a violation to the absence of arbitrage principle which persists under small transaction costs.
This result is in line with the empirical evidence reported in \cite{Ted} (compare with Remark \ref{rem:Ted}).
We also want to point out that the same reasoning applies to Proposition \ref{prop:RC}, thus implying that the necessity of right-continuity is robust with respect to small transaction costs.
Furthermore, an argument similar to that given in the proof of Proposition \ref{prop:robust} allows to show that constant profits via flash strategies are robust with respect to small fixed (instead of proportional)  transaction costs.

\begin{rem}
In general, Proposition \ref{prop:robust} cannot be extended to flash strategies generating sure (but not constant) profits. As a simple counterexample, consider the process $X:=1+\eta\ind_{\dbraco{1,+\infty}}$ in its natural filtration $\FF$, where $\eta$ is for instance an exponential random variable. Obviously, $X$ admits sure profits via flash strategies.
For $\varepsilon>0$, let $\Xtilde:=1+\varepsilon+(\eta-\varepsilon)\ind_{\dbraco{1,+\infty}}$, which is $\varepsilon$-close to $X$. However, $\Xtilde$ does not admit sure profits via flash strategies, since $\{\Delta\Xtilde_1>0\}=\{\eta>\varepsilon\}\notin\cF_{1-}=\cF_0$.
The financial intuition is that, if the direction of a jump is known but its size is unpredictable, then the profits generated by a flash strategy may not suffice to compensate the transaction costs incurred.
\end{rem}

\subsection{The semimartingale case}	\label{sec:semimg}

Theorem \ref{thm:main} holds true for any c\`adl\`ag adapted process $X$. If in addition $X$ is assumed to be a semimartingale, then the absence of (fully) predictable jumps admits a further simple characterization. For a semimartingale $X$, we say that a bounded predictable process $h$ is an {\em instantaneous strategy} if it is of the form $h=\xi\ind_{\dbra{\tau}}$, for some bounded random variable $\xi$ and a stopping time $\tau$. In the spirit of Definition \ref{def:flash}, we say that an instantaneous strategy $h$ generates {\em  sure profits} if $(h\cdot X)_t=\zeta\,\ind_{\{\tau\leq t\}}$ a.s. for every $t\geq0$, for some random variable $\zeta$ such that $\{\tau<+\infty\}\subseteq\{\zeta>0\}$. 
If $\PP(\zeta=c)=1$, for some constant $c>0$, then the instantaneous strategy $h$ is said to generate a {\em constant profit}.

\begin{cor}	\label{cor:semimg}
Assume that the process $X$ is a semimartingale. Then the following are equivalent:
\begin{enumerate}[(i)]
\item $X$ does not exhibit predictable (fully predictable, resp.) jumps;
\item there are no sure (constant, resp.) profits via flash strategies;
\item there are no sure (constant, resp.) profits via instantaneous strategies.
\end{enumerate}
\end{cor}
\begin{proof}
For brevity, we shall only consider the cases of predictable jumps and sure profits.
$(ii)\Rightarrow(i)$: 
this implication follows from Theorem~\ref{thm:main}.
$(iii)\Rightarrow(ii)$:
let $(h^n)_{n\in\N}$ be a flash strategy, composed of elements $h^n=\xi^n\ind_{\dbraoc{\sigma_n,\tau_n}}$, such that $\lim_{n\rightarrow+\infty}(h^n\cdot X)_t=\zeta\,\ind_{\{\tau\leq t\}}$ a.s. for every $t\geq0$, for some random variable $\zeta$ and a stopping time $\tau$ with $\{\tau<+\infty\}\subseteq\{\zeta>0\}$. As shown in the second part of the proof of Theorem \ref{thm:main}, $h^n$ converges a.s. to $h:=\xi\ind_{\dbra{\tau}}$ for $n\rightarrow+\infty$. The dominated convergence theorem for stochastic integrals (see \cite[Theorem IV.32]{MR1037262}) then implies that 
$(h\cdot X)_t=\lim_{n\rightarrow+\infty}(h^n\cdot X)_t=\zeta\,\ind_{\{\tau\leq t\}}$ a.s., for every $t\geq0$.
$(i)\Rightarrow(iii)$:
let $h=\xi\ind_{\dbra{\tau}}$ be an instantaneous strategy generating sure profits. By \cite[\textsection~I.4.38]{MR1943877}, it holds that $h\cdot X=\xi\Delta X_{\tau}\ind_{\dbraco{\tau,+\infty}}$, so that $\xi\Delta X_{\tau}=\zeta>0$ a.s. on $\{\tau<+\infty\}$. This implies that $\tau$ is a jump time of $X$.
Moreover, it holds that $\{\tau<+\infty,\Delta X_{\tau}>0\}=\{\tau<+\infty,\xi>0\}\in\cF_{\tau-}$, due to the predictability of $h$ and since $\{\tau<+\infty\}\subseteq\{\zeta>0\}$.
Finally, the predictability of $\tau$ follows by noting that $\dbra{\tau}=\{h\neq 0\}$.
\end{proof}

In the proof of Corollary \ref{cor:semimg}, the semimartingale property is used to ensure that the gains from trading generated by a sequence of buy-and-hold strategies forming a flash strategy converge to the gains from trading generated by an instantaneous strategy.

In the semimartingale case, under the additional assumption of quasi-left-continuity of the filtration $\FF$ (i.e., $\cF_T=\cF_{T-}$ for every predictable time $T$, see \cite[Section IV.3]{MR1037262}), it has been shown in \cite{H85} that predictable jumps cannot occur if the financial market is viable, in the sense that there exists an optimal consumption plan for some agent with sufficiently regular preferences. 
In our context, this result is a direct consequence of Corollary \ref{cor:semimg}, as the existence of sure profits is clearly incompatible with any form of market viability, regardless of the quasi-left-continuity of $\FF$.

\subsection{Comparison with other no-arbitrage conditions}	\label{sec:relations}

The absence of sure profits from instantaneous strategies must be regarded as a minimal no-arbitrage condition. In particular, in the semimartingale case, it is implied by the requirement of {\em no increasing profit} (NIP), itself an extremely weak no-arbitrage condition for a financial model (see \cite{F15} and \cite[Section 3.4]{MR2335830})\footnote{An $X$-integrable predictable process $h$ is said to generate an {\em increasing profit} if the gains from trading process $h\cdot X$ is predictable, non-decreasing and satisfies $\PP((h\cdot X)_T>0)>0$ for some $T>0$, see \cite[Definition 2.2]{F15}.}.

The absence of predictable jumps can be directly proven by martingale methods under the classical {\em no free lunch with vanishing risk} (NFLVR) condition. Note, however, that NFLVR is much stronger than the absence of sure profits as considered above. We recall that NFLVR is equivalent to the existence of a probability measure $\QQ\sim\PP$ such that $X$ is a sigma-martingale under $\QQ$ (see \cite{MR1671792}).
For completeness, we present the following proposition with its  simple proof.

\begin{prop}	\label{prop:NFLVR}
Assume that the process $X$ is a semimartingale satisfying NFLVR. Then $X$ cannot exhibit predictable jumps.
\end{prop}
\begin{proof}
If $X$ satisfies NFLVR, then there exists $\QQ\sim\PP$ and an increasing sequence of predictable sets $(\Sigma_k)_{k\in\N}$ with $\bigcup_{k\in\N}\Sigma_k=\Omega\times\R_+$ such that $\ind_{\Sigma_k}\cdot X$ is a uniformly integrable martingale under $\QQ$, for every $k\in\N$ (see \cite[Definition III.6.33]{MR1943877}). 
Let $T$ be a predictable time such that $\dbra{T}\subseteq\{\Delta X\neq0\}$ and $\ind_{\{T<+\infty,\,\Delta X_T>0\}}$ is $\cF_{T-}$-measurable. 
With the convention $\Delta X_T=0$ on $\{T=+\infty\}$, we define the predictable times $\tau^1:=T\ind_{\{\Delta X_T>0\}}+\infty\ind_{\{\Delta X_T\leq 0\}}$ and $\tau^2:=T\ind_{\{\Delta X_T<0\}}+\infty\ind_{\{\Delta X_T\geq 0\}}$. 
Hence,
\begin{align*}
Y^{(k)} :&= \ind_{\Sigma_k}\cdot\left(|\Delta X_T|\ind_{\dbraco{T,+\infty}}\right)
= \ind_{\Sigma_k}\cdot\left(\Delta X_{\tau^1}\ind_{\dbraco{\tau^1,+\infty}}-\Delta X_{\tau^2}\ind_{\dbraco{\tau^2,+\infty}}\right)	\\
&= \ind_{\Sigma_k}\cdot\left((\ind_{\dbra{\tau^1}}-\ind_{\dbra{\tau^2}})\cdot X\right)
= (\ind_{\dbra{\tau^1}}-\ind_{\dbra{\tau^2}})\cdot\left(\ind_{\Sigma_k}\cdot X\right),
\end{align*}
for every $k\in\N$, where we have used \cite[\textsection~I.4.38]{MR1943877} and the associativity of the stochastic integral. 
Therefore, the process $Y^{(k)}$ is a non-decreasing local martingale. Since $Y^{(k)}_0=0$, this implies that $Y^{(k)}\equiv 0$ (up to an evanescent set), for all $k\in\N$.
In turn, this implies that $|\Delta X_T|=0$ a.s., contradicting the assumption that $T$ is a jump time of $X$.
\end{proof}


\begin{rem}	\label{rem:NUPBR}
Predictable jumps can never occur under the {\em no unbounded profit with bounded risk} (NUPBR) condition, introduced in \cite[Definition 4.1]{MR2335830}. This follows by Proposition \ref{prop:NFLVR}, noting that NUPBR is equivalent to NFLVR up to a localizing sequence of stopping times.
In turn, since NUPBR is equivalent to existence and finiteness of the growth-optimal portfolio (see \cite[Theorem 4.12]{MR2335830}), this implies that predictable jumps are always excluded in the context of the {\em benchmark approach} (see \cite{PH}).
\end{rem}


\section{Conclusions}	\label{sec:conclusions}

In this paper we have shown that, under minimal assumptions, the possibility of realizing sure (constant, resp.) profits via flash strategies is equivalent to the existence of jumps of predictable direction (direction and magnitude, resp.) occurring at predictable times. 
Excluding sure profits via flash strategies, we have also shown that right-continuity is an indispensable path property for any asset price process. 
Since flash strategies represent well typical strategies adopted by high-frequency traders, as explained in the introduction, we deduce that the profitability of high-frequency strategies is closely related to the presence of information not yet incorporated in market prices. 
In this sense, the arbitrage activity of high-frequency traders should have a beneficial role in price discovery and lead to an increase of market efficiency (see \cite{HS13,BHR} for empirical results in this direction). 
However, a general analysis of the impact of high-frequency trading is definitely beyond the scope of this paper.

Finally, we want to emphasize that, since the notion of predictability depends on the reference filtration, the possibility of realizing sure profits via flash strategies depends on the information set under consideration. This means that financial markets can be efficient in the semi-strong form and sure profits via flash strategies impossible to achieve for ordinary investors having access to publicly available information, while investors having access to privileged information ({\em insider traders}) can have an information set rich enough to allow for sure profits via flash strategies, so that market efficiency does not hold in the strong form. This is simply a consequence of Theorem \ref{thm:main} together with the fact that the predictable sigma-field associated to a smaller filtration is a subset of the predictable sigma-field associated to a larger filtration.
This observation is in line with the empirical literature documenting violations to strong-form market efficiency in the presence of insider information (see e.g. \cite[Section 6]{Fama2}).
This is also in line with the empirical analysis of \cite{HLS}, where it is shown that institutional traders have an informational advantage which allows to predict to some extent the time and the content of news announcements as well as the returns on the announcement date.
Furthermore, informed trading represents one of the sources of the profits of high-frequency strategies, as high-frequency traders have access to information which is not available to ordinary market participants. This information-based explanation of high-frequency profits has been recently addressed in \cite{JL12} and \cite{KP15}.

\bibliographystyle{alpha}
\bibliography{biblio_jumps}

\begin{thebibliography}{ABD10}

\bibitem[ABD10]{ABD}
T.G. Andersen, T.~Bollerslev, and F.X. Diebold.
\newblock Parametric and nonparametric volatility measurement.
\newblock In Y.~A\"it-Sahalia and L.P. Hansen, editors, {\em Handbook of
  Financial Econometrics}, volume~1, chapter~2, pages 67--137. North-Holland,
  2010.

\bibitem[Ald13]{Aldridge}
I.~Aldridge.
\newblock {\em {H}igh-{F}requency {T}rading: {A} {P}ractical {G}uide to
  {A}lgorithmic {S}trategies and {T}rading {S}ystems}.
\newblock Wiley, Hoboken (NJ), second edition, 2013.

\bibitem[Bat03]{Battauz}
A.~Battauz.
\newblock Quadratic hedging for asset derivatives with discrete stochastic
  dividends.
\newblock {\em Insur. Math. Econ.}, 32(2):229--243, 2003.

\bibitem[BHR14]{BHR}
J.~Brogaard, T.~Hendershott, and R.~Riordan.
\newblock High-frequency trading and price discovery.
\newblock {\em Rev. Financ. Stud.}, 27(8):2267--2306, 2014.

\bibitem[CT15]{CT15}
H.N. Chau and P.~Tankov.
\newblock Market models with optimal arbitrage.
\newblock {\em SIAM J. Financ. Math.}, 6:66--85, 2015.

\bibitem[DJ05]{DubJoh}
A.~Dubinsky and M.~Johannes.
\newblock Earnings announcements and equity options.
\newblock Columbia University, working paper, 2005.

\bibitem[DS98]{MR1671792}
F.~Delbaen and W.~Schachermayer.
\newblock The fundamental theorem of asset pricing for unbounded stochastic
  processes.
\newblock {\em Math. Ann.}, 312(2):215--250, 1998.

\bibitem[Eva11]{Evans}
K.P. Evans.
\newblock Intraday jumps and {US} macroeconomic news announcements.
\newblock {\em J. Bank. Financ.}, 35(10):2511--2527, 2011.

\bibitem[Fam70]{Fama}
E.F. Fama.
\newblock Efficient capital markets: A review of theory and empirical work.
\newblock {\em J. Financ.}, 25(2):383--417, 1970.

\bibitem[Fam91]{Fama2}
E.F. Fama.
\newblock Efficient capital markets: {II}.
\newblock {\em J. Financ.}, 46(5):1575--1617, 1991.

\bibitem[Fon15]{F15}
C.~Fontana.
\newblock Weak and strong no-arbitrage conditions for continuous financial
  markets.
\newblock {\em Int. J. Theor. Appl. Finan.}, 18(1):1550005, 2015.

\bibitem[FS18]{FS18}
C.~Fontana and T.~Schmidt.
\newblock General dynamic term structures under default risk.
\newblock {\em Stoch. Proc. Appl.}, 128(10):3353--3386, 2018.

\bibitem[Ges79]{Geske}
R.~Geske.
\newblock A note on an analytical valuation formula for unprotected {A}merican
  call options on stocks with known dividends.
\newblock {\em J. Financ. Econ.}, 7(4):375--380, 1979.

\bibitem[GR15]{GR15}
P.~Guasoni and M.~R\'asony.
\newblock Fragility of arbitrage and bubbles in local martingale diffusion
  models.
\newblock {\em Finance Stoch.}, 19(2):215--231, 2015.

\bibitem[GS18]{GS16}
F.~Gehmlich and T.~Schmidt.
\newblock Dynamic defaultable term structure modeling beyond the intensity
  paradigm.
\newblock {\em Math. Financ.}, 28(1):211--239, 2018.

\bibitem[HJ88]{HJ88}
D.~Heath and R.~Jarrow.
\newblock Ex-dividend stock price behavior and arbitrage opportunities.
\newblock {\em J. Bus.}, 61(1):95--108, 1988.

\bibitem[HLS15]{HLS}
T.~Hendershott, D.~Livdan, and N.~Sch\"urhoff.
\newblock Are institutions informed about news?
\newblock {\em J. Financ. Econ.}, 117(2):249--287, 2015.

\bibitem[HS13]{HS13}
J.~Hasbrouck and G.~Saar.
\newblock Low-latency trading.
\newblock {\em J. Financ. Mark.}, 16(4):646--679, 2013.

\bibitem[Hua85]{H85}
C.-F. Huang.
\newblock Information structures and viable price systems.
\newblock {\em J. Math. Econ.}, 14(3):215--240, 1985.

\bibitem[JL12]{JL12}
R.~Jarrow and H.~Li.
\newblock Abnormal profit opportunities and the informational advantage of high
  frequency trading.
\newblock {\em Q. J. Financ.}, 3(2):1350012, 2012.

\bibitem[JS03]{MR1943877}
J.~Jacod and A.N. Shiryaev.
\newblock {\em Limit Theorems for Stochastic Processes}.
\newblock Springer, Berlin, second edition, 2003.

\bibitem[KK07]{MR2335830}
I.~Karatzas and C.~Kardaras.
\newblock The num\'eraire portfolio in semimartingale financial models.
\newblock {\em Finance Stoch.}, 11(4):447--493, 2007.

\bibitem[KP15]{KP15}
Y.~Kchia and P.~Protter.
\newblock Progressive filtration expansion via a process, with applications to
  insider trading.
\newblock {\em Int. J. Theor. Appl. Finan.}, 18(4):155027, 2015.

\bibitem[KV91]{KV}
O.~Kim and R.E. Verrecchia.
\newblock Market reaction to anticipated announcements.
\newblock {\em J. Financ. Econ.}, 30(2):273--309, 1991.

\bibitem[KW14]{KimWright}
D.H. Kim and J.H. Wright.
\newblock Jumps in bond yields at known times.
\newblock Federal Reserve Board, Washington, discussion paper, 2014.

\bibitem[Lee12]{Lee}
S.S. Lee.
\newblock Jumps and information flow in financial markets.
\newblock {\em Rev. Financ. Stud.}, 25(2):439--479, 2012.

\bibitem[Lew14]{Lewis}
M.~Lewis.
\newblock {\em Flash Boys: A Wall Street Revolt}.
\newblock Norton \& Company, New York - London, 2014.

\bibitem[LM08]{LeeMyk}
S.S. Lee and P.A. Mykland.
\newblock Jumps in financial markets: a new nonparametric test and jump
  dynamics.
\newblock {\em Rev. Financ. Stud.}, 21(6):2535--2563, 2008.

\bibitem[PH06]{PH}
E.~Platen and D.~Heath.
\newblock {\em A Benchmark Approach to Quantitative Finance}.
\newblock Springer, Berlin - Heidelberg, 2006.

\bibitem[Pia01]{piaz01}
M.~Piazzesi.
\newblock An econometric model of the yield curve with macroeconomic jump
  effects.
\newblock NBER working paper no 8246, 2001.

\bibitem[Pia05]{Piazzesi}
M.~Piazzesi.
\newblock Bond yields and the {F}ederal {R}eserve.
\newblock {\em J. Polit. Econ.}, 113(2):311--344, 2005.

\bibitem[Pro04]{MR1037262}
P.~Protter.
\newblock {\em Stochastic Integration and Differential Equations}.
\newblock Springer, Berlin - Heidelberg, 2.1 edition, 2004.

\bibitem[Pro15]{Pro_review}
P.~Protter.
\newblock Flash boys: cracking the money code, by {M}ichael {L}ewis.
\newblock {\em Quant. Financ.}, 15(2):205--206, 2015.

\bibitem[Ran11]{Rangel}
J.G. Rangel.
\newblock Macroeconomic news, announcements, and stock market jump intensity
  dynamics.
\newblock {\em J. Bank. Financ.}, 35(5):1263--1276, 2011.

\bibitem[Rol77]{Roll}
R.~Roll.
\newblock An analytic valuation formula for unprotected {A}merican call options
  on stocks with known dividends.
\newblock {\em J. Financ. Econ.}, 5(2):251--258, 1977.

\bibitem[Ted17]{Ted}
D.~Tedeschini.
\newblock Approximate arbitrage with limit orders.
\newblock Working paper, Swiss Finance Institute, 2017.

\bibitem[Wha81]{Whaley}
R.E. Whaley.
\newblock On the valuation of {A}merican call options on stocks with known
  dividends.
\newblock {\em J. Financ. Econ.}, 9(2):207--211, 1981.

\end{thebibliography}
\end{document}